\documentclass[11pt,a4paper]{article}

\usepackage[hidelinks]{hyperref}
\usepackage{amssymb, amsmath, graphicx, subcaption}
\usepackage{multirow}
\usepackage{float}
\usepackage{fancyhdr}
\usepackage{amsthm}

\newtheorem{theorem}{Theorem}
\newtheorem{lemma}{Lemma}
\newtheorem{corollary}{Corollary}
\newtheorem{definition}{Definition}

\fancypagestyle{firststyle}
{
	\fancyhf{}
	\fancyfoot[LE,RO]{Page \thepage}
}

\title{Non-self-intersecting trajectories and their applications to satellite constellation design and orbital capacity}

\author{David Arnas\thanks{Purdue University, IN, USA. Email: \textsc{darnas@purdue.edu}}, Richard Linares\thanks{Massachusetts Institute of Technology, MA, USA. Email: \textsc{linaresr@mit.edu}}}

\begin{document}
	
	\date{}	
	
	\maketitle
	
	\thispagestyle{firststyle}

\begin{abstract}
    This work focuses on the generation of non-self-intersecting relative trajectories, and their applications to satellite constellation design, slotting architectures, and Space Traffic Management. To that end, this paper introduces two theorems to determine when two spacecrafts share the same relative trajectory, and to identify the only conditions that allow the existence of non-self-intersecting relative trajectories. Then, these results are applied first to the estimation of the limits of the orbital capacity at a given altitude, and second, to the design of satellite constellations and slotting architectures that present no conjunctions between any element compliant with these space structures.
\end{abstract}

\section{Introduction}

In the last decade we have witnessed a dramatic increase on the number of objects orbiting the Earth, especially in Low Earth Orbits (LEO) due to its interesting applications for telecommunications and Earth observation. This situation has led to an increase on the number of orbital conjunctions between satellites and thus the safety of not only current but also future missions. On top of that, requiring to perform additional collision avoidance maneuvers translates into a reduction of the fuel budget for space missions and an increase on the costs for controlling spacecraft in these congested regions. Unfortunately, this is a trend that is expected to continue in the future with more spacecraft and satellite constellations planned to be launched on the following years.

In this regard, a wide variety of studies have been performed to assess orbital conjunctions~\cite{alfano,conjunction} and orbital risk~\cite{risk_patera,risk_alfano}, as well as different control techniques to perform collision avoidance maneuvers~\cite{CAM_formation,opt_impulsive,opt_low}. In here instead, we follow the idea from Ref.~\cite{stm} to create a system level solution based on defining slotting architectures that ensure that no conjunction will happen between satellites compliant with the defined structure. To that end, we make use of analytical constellation theory to define the positions of these slots and to analyze the results of the resultant distributions.

Analytical constellation design formulations have been and are extensively used to define and analyze the complex problem that represents satellite constellation design. Examples of this kind of formulations include, for instance, Walker Constellations~\cite{WalkerSC_1}, Draim Constellations~\cite{Draim4}, Dufour Constellations~\cite{dufour}, Rosette Constellations~\cite{Ballard}, Kinematically Regular Satellite Networks~\cite{Mozhaev}, and Flower Constellations~\cite{MortariFC}. From them, it is of special interest the case of Flower Constellations, specially the so called Lattice Flower Constellations~\cite{2dlfc,3dlfc,4dlfc,ndnfc}, as they represent the generalization of these analytical constellation formulations based on uniform distributions. In particular, 2D Lattice Flower Constellations was the formulation proposed to generate optimal slotting architectures in Ref.~\cite{stm}. 

The objective of this work is to use analytical constellation theory to define satellite constellations and slotting architectures in circular orbits that present no conjunctions between compliant elements from the structure. In this sense, this work aims to complement and extend the results presented in Ref.~\cite{stm} for the Space Traffic Management problem by proposing a design framework to analytically estimate the maximum capacity of these systems while taking into account the minimum distance between satellites. To that end, this work focuses on the generation of satellite constellations based on circular orbits where spacecrafts are distributed following a single relative trajectory that present no self-intersections in a given rotating frame of reference. This means that no matter where satellites are located along this relative trajectory, there is no risk of conjunction since the constellation satellites follow the same relative path with a given delay between satellites. This idea was initially explored in Ref~\cite{largeconstellations}, however, results were not complete and their design was limited to a specific range of inclinations. In here, we aim to extend and generalize these results by providing a set of theorems that define the conditions that allow a non-self-intersecting trajectory to exist, and show how to apply this result to satellite constellation design, slotting architectures, and the estimation of orbital capacity. 

This work is organized as follows. First, a summary of the satellite constellation design formulations used in this work are presented. Particularly, this work is based on the formulation introduced in Ref.~\cite{Time} in order to define constellations following a set of relative space-tracks. Second, the conditions to generate a relative trajectory with non-self-intersections are introduced. This includes the introduction of a general theorem to determine if two satellites belong to the same relative trajectory, a theorem to identify the only conditions that allow a relative trajectory to have non-self-intersections, and a lemma that provides a closed form approximation to the previous theorem. Third, this design of non-self-intersecting trajectories is applied to satellite constellation design, and more in particular, to the efficient computation of minimum distance between satellites in these structures and the analytical evaluation of the capacity of these systems. Examples of application are also provided.


\section{Preliminaries}

This section includes a summary of some analytical satellite constellation design formulations that will be used in this manuscript. Particularly, we include the time distribution formulation proposed in Ref.~\cite{Time} as it is the basis used to define and study non-self-intersecting relative trajectories and constellations. In addition, the formulations of Walker Constellations and 2D Lattice Flower Constellations are included due to their historical importance and the fact that they generalize the problem of uniform distribution of satellites in space.

\subsection{Time distributions in satellite constellation design}

Usually, satellite constellation designs as Walker Constellations or 2D Lattice Flower Constellations define their satellite distribution in the inertial frame of reference. However, there are some problems where it is more interesting to study the system from the point of view of a rotating frame of reference. Examples of that include, for instance, the definition of the ground-track distribution of the constellation for Earth observation satellites~\cite{Time2}, computing the revisiting time~\cite{maintenance}, or including orbital perturbations in the nominal design of the constellation~\cite{nominal}. To that end, Ref.~\cite{Time} proposed an analytical formulation to perform the satellite distribution definition in the rotating frame of reference instead that in the inertial frame of reference. 

The idea behind this satellite constellation design methodology is to define a set of relative trajectories in a given rotating frame of reference (such as the one associated with the Earth) and use the relative along track time distance ($t_{kq}$) and cross track angular distances between relative trajectories ($\Delta\Omega_k$) to establish the relative distribution of the constellation. Particularly, all the satellites in the constellation share the same values of semi-major axis ($a$), eccentricity ($e$), inclination ($inc$), and argument of perigee ($\omega$) while the distribution is performed in the right ascension of the ascending node ($\Omega$) and mean anomaly ($M$) of the orbits. Let $\alpha$ be the spin rate of the rotating frame of reference about the z axis, where $\alpha$ is positive if the rotation is towards the East, and negative if it is towards the West. Then, the relative satellite distribution can be defined as~\cite{Time}:
\begin{eqnarray} \label{eq:apgtc}
\Delta \Omega_{kq} & = & \Delta\Omega_k - \alpha(t_{kq}-t_0), \nonumber \\
\Delta M_{kq} & = & n(t_{kq}-t_0),
\end{eqnarray}
where $\Delta \Omega_{kq}$ and $\Delta M_{kq}$ describe the relative position of each spacecraft with respect to a reference position of the constellation defined by $t_0$, $n$ is the mean motion of the satellites, and $\{k,q\}$ names each satellite in the relative track $k$ and position $q$ of the relative track.

In this work we are interested in uniform distributions over a single relative trajectory that is closed in the chosen rotating reference frame. This implies that $\Delta\Omega_k = 0$ for all the satellites of the constellation. Moreover, since the relative trajectory is closed, a cycle period $T_c$ can be defined such that:
\begin{equation}\label{eq:cycle}
T_c = N_pT = N_dT_d,
\end{equation}
where $N_d$ and $N_p$ are the number of complete rotations of the satellite and the rotating reference frame respectively to repeat the cycle. Note that $N_d$ and $N_p$ are coprime integer numbers. Therefore, the former equation can be rewritten as:
\begin{equation}\label{eq:cycle2}
T_c = N_p\displaystyle\frac{2\pi}{n} = N_d\frac{2\pi}{\alpha}.
\end{equation}
Additionally, and since we are interested in uniform distributions, we can define the time distribution of satellites as:
\begin{equation}
\displaystyle\frac{q}{N_s} = \frac{t_q - t_0}{T_c}, 
\end{equation}
where $q\in\{1,\dots,N_s\}$ names each of the $N_s$ satellites of the constellation. Thus, we can define the uniform distribution in a single relative trajectory through the following expression:
\begin{eqnarray} \label{eq:tc_single}
\Delta \Omega_{q} & = & - 2\pi N_d\displaystyle\frac{q}{N_s}, \nonumber \\
\Delta M_{q} & = & 2\pi N_p\displaystyle\frac{q}{N_s}.
\end{eqnarray}

\subsection{Walker Constellations}

Walker-Delta Constellations~\cite{WalkerSC_1} are the most used satellite constellation design both in the literature and industry. They are based on the idea of generating an uniform distribution of satellites based on circular orbits, where the distribution itself in defined in the mean anomaly and right ascension of the ascending node of the spacecraft's orbits. This means that all satellites in the constellation share the same semi-major axis, eccentricity, inclination, and argument of perigee. 

A Walker Constellation can be defined in a compact form as: $inc: \ t/y/f$, where $inc$ is the inclination of the constellation, $t$ is the number of satellites, $y$ is the number of orbital planes, and $f \in \{0, \ldots, y-1\}$ is the distribution parameter that controls the relative phasing in mean anomaly between two consecutive orbital planes from the constellation. In other words, the distribution of satellites can be expressed as:
\begin{eqnarray}\label{eq:walker0}
	\Delta\Omega_{ij} & = & 2\pi\displaystyle\frac{i}{y}, \nonumber \\ 
	\Delta M_{ij} & = & 2\pi\displaystyle\frac{y}{t}j + 2\pi\displaystyle\frac{f}{t}i,
\end{eqnarray}
where $\Delta \Omega_{ij}$ and $\Delta M_{ij}$ define the relative position of the spacecraft with respect to a reference satellite of the constellation, and $i$ and $j$ name each of the satellites located in the orbit $i$ and position $j$ of that orbit.

\subsection{2D Lattice Flower Constellations}

2D Lattice Flower Constellations~\cite{2dlfc} is a satellite constellation design formulation, based on the properties of Number Theory, that performs uniform distributions of satellites in the right ascension of the ascending node and the mean anomaly of the orbits. As in the previous formulations, the values of semi-major axis, eccentricity, inclination, and argument of perigee are common for all the satellites of the constellation. 2D Lattice Flower Constellations are defined using three integer parameters, the number of orbits $N_{o}$, the number of satellites per orbit $N_{so}$, and the combination number $N_{c}$, and the following relative distribution of satellites $\{\Delta \Omega_{ij}, \Delta M_{ij}\}$:
\begin{eqnarray}\label{eq:2dlfc}
\Delta\Omega_{ij} & = & \displaystyle\frac{2\pi}{N_{o}}i, \nonumber \\ 
\Delta M_{ij} & = & \displaystyle\frac{2\pi}{N_{so}}\Big(j - \frac{N_{c}}{N_{o}}i\Big),
\end{eqnarray}
where $i\in\{1,\dots,N_{o}\}$ and $j\in\{1,\dots,N_{so}\}$ names each of the satellites of the constellation in their orbital planes and their position in the orbit respectively. Note that this formulation is equivalent to the one of Walker Constellations, where the combination number ($N_{c}$) is this time the parameter that controls the relative phasing between two consecutive orbital planes from the configuration. In fact, 2D Flower Constellations are the generalization of Walker Constellations and can be applied to both circular and elliptic orbits. The relation between the parameters from Walker and 2D Lattice Flower Constellations are: $t = N_{o}N_{so}$, $y = N_{o}$, and $f = -N_{c} \mod(L_{\Omega})$. For this reason, we will focus on this work on the 2D Lattice Flower Constellation formulation since it is a more general formulation.


\section{Non-self-intersecting trajectories}

As a first step in the design of this space architectures, it is necessary to find the conditions that assure a non-self-intersecting relative trajectory. To that end, two theorems and a lemma are proposed in the following lines that represent the generalization and extension of the partial solution presented in Ref.~\cite{largeconstellations}. In that regard, Theorem~\ref{theorem:compatibility} proves the condition for two satellites at the same inclination and semi-major axis to share the same relative trajectory. Theorem~\ref{theorem:nonselfintersectioin}, on the other hand, presents the constraints that orbital parameters and the rotating frame have to fulfill in order to have a non-self-intersecting relative trajectory. And finally, Lemma~\ref{lemma:approx} contains the analytical close form approximation of the result obtained in Theorem~\ref{theorem:nonselfintersectioin}.

\begin{theorem} \label{theorem:compatibility}
	Two satellites are located in the same relative trajectory defined in a reference frame, that is rotating at a constant spin rate of $\alpha$ about the $z$ axis, if and only if their semi-major axis, inclination, eccentricity and argument of perigee are equal, and their relative position fulfills the following expression:
	\begin{equation}
	N_p\Delta\Omega \pm N_d\Delta M = 0 \mod(2\pi),
	\end{equation}
	where the $+$ sign is used in prograde rotating reference frames ($\alpha$ is positive), and the $-$ sign for retrograde frames ($\alpha$ is negative).
\end{theorem} 

\begin{proof}
	The instantaneous relative trajectory of a given satellite in a reference frame rotating at an angular velocity $\alpha$ about the $z$ axis can be defined using Eq.~\eqref{eq:apgtc}:
	\begin{eqnarray} \label{eq:single_trajectory}
	\Delta \Omega & = & - \alpha t, \nonumber \\
	\Delta M & = & n t,
	\end{eqnarray}
	where $t\in[0,T_c)$ is the variable that allows to describe the trajectory. Note that all the positions defined by this relative trajectory share the same semi-major axis, inclination, eccentricity and argument of perigee of the original satellite since they were defined by the propagation of this spacecraft. Additionally, due to the fact that $\Delta \Omega$ and $\Delta M$ are angles, they are subjected to modular arithmetic, and thus, the previous expressions can be rewritten as:
	\begin{eqnarray}
	\Delta \Omega & = & - \alpha t + 2\pi G, \nonumber \\
	\Delta M & = & n t + 2\pi H,
	\end{eqnarray}
	where $G$ and $H$ are two arbitrary integer numbers. Both expressions can be combined by removing the effect of the time distribution $t$, obtaining:
	\begin{equation}
	\Delta\Omega = -\displaystyle\frac{\alpha}{n}\left(\Delta M - 2\pi H\right) + 2\pi G.
	\end{equation}
	On the other hand, from Eq.~\eqref{eq:cycle2}:
	\begin{equation} \label{eq:ndnp}
	\displaystyle\frac{\alpha}{n} = \pm \frac{N_d}{N_p},
	\end{equation}
	where the $+$ sign of $\pm$ represents prograde rotating frames (as in the case of the rotating frame associated with the Earth), while the $-$ relates to retrograde frames. Note that the inclusion of $\pm$ in the equation is required due to the fact that $N_d$ and $N_p$ are defined as positive integers (they represent number of complete revolutions). Therefore:
	\begin{equation}
	\Delta\Omega = \mp\displaystyle\frac{N_d}{N_p}\left(\Delta M - 2\pi H\right) + 2\pi G,
	\end{equation}
	which, after rearranging the expression, we obtain:
	\begin{equation}
	N_p\Delta\Omega \pm N_d \Delta M = 2\pi\left(G \pm  H\right),
	\end{equation}
	where the sum $G \pm  H$ can generate any integer number (remember that $G$ and $H$ are two free arbitrary integers), therefore, the expression can be rewritten as:
	\begin{equation}
	N_p\Delta\Omega \pm N_d \Delta M = 0 \mod\left(2\pi\right).
	\end{equation}
\end{proof}

Note that this expression represents the generalization of the compatibility equation presented in Ref.~\cite{2dlfc} where the consideration of retrograde rotating reference frames has been taken into account. We will now use this generalization to also extend the solution from Ref.~\cite{largeconstellations} to study circular non-self-intersecting trajectories that account for both prograde and retrograde reference frames.

\begin{theorem}\label{theorem:nonselfintersectioin}
	The relative trajectory described by a satellite in a circular orbit has no self intersection if and only if one of these for cases applies:
	\begin{equation}
	\left\{
	\begin{tabular}{l}
	$\alpha > 0$ and $N_p = N_d - 1$: $\cos(inc) > \displaystyle\frac{N_p}{N_d},$ \nonumber \\
	$\alpha > 0$ and $N_p = N_d + 1$: $\cos(inc) > \max\left(\displaystyle\frac{\tan(\pi N_p \tau)}{\tan(\pi N_d \tau)}\right),$ \nonumber \\
	$\alpha < 0$ and $N_p = N_d - 1$: $\cos(inc) < -\displaystyle\frac{N_p}{N_d},$ \nonumber \\
	$\alpha < 0$ and $N_p = N_d + 1$: $\cos(inc) < -\max\left(\displaystyle\frac{\tan(\pi N_p \tau)}{\tan(\pi N_d \tau)}\right),$
	\end{tabular}
	\right.
	\end{equation}
	where $\tau$ is defined in the range $(\frac{1}{N_d+N_p},\frac{3}{2}\frac{1}{N_d+N_p}]$.	
\end{theorem}

\begin{proof}
	The goal is to find a condition in a relative trajectory such that the minimum distance between two points from that relative trajectory is never zero.	From Ref.~\cite{mindis}, the minimum normalized distance between two satellites in circular orbits and at the same inclination is:
	\begin{eqnarray}\label{eq:midis}
	\rho_{min} & = & 2\left|\displaystyle\sqrt{\displaystyle\frac{1+\cos^2(inc)+\sin^2(inc)\cos(\Delta\Omega)}{2}}\sin\left(\frac{\Delta F}{2}\right)\right|; \nonumber \\
	\Delta F & = & \Delta M - 2\arctan\left(-\cos(inc)\tan\left(\displaystyle\frac{\Delta\Omega}{2}\right)\right). 
	\end{eqnarray}
	Moreover, the trajectory described by a satellite in the selected rotating reference frame is defined by Equation~\eqref{eq:single_trajectory}. Let the non-dimensional distribution variable $\tau$ be defined as $\tau = t/T_c$ where $\tau\in[0,1)$, then, the trajectory can be defined as:
	\begin{eqnarray} 
	\Delta \Omega & = & - \alpha \tau T_c, \nonumber \\
	\Delta M & = & n \tau T_c,
	\end{eqnarray}
	and using Eq.~\eqref{eq:cycle2}:
	\begin{eqnarray} \label{eq:plusminusdistribution}
	\Delta \Omega & = & \mp 2\pi N_d \tau, \nonumber \\
	\Delta M & = & 2\pi N_p \tau,
	\end{eqnarray}
	where the upper sign in $\mp$ accounts for prograde rotating reference frames ($\alpha >0$), while the lower sign for retrograde frames ($\alpha < 0$). We substitute these values back into Eq.~\eqref{eq:midis} to obtain:
	\begin{eqnarray}
	\rho_{min} & = & 2\left|\displaystyle\sqrt{\displaystyle\frac{1+\cos^2(inc)+\sin^2(inc)\cos(\mp 2\pi N_d \tau)}{2}}\sin\left(\frac{\Delta F}{2}\right)\right|; \nonumber \\
	\Delta F & = & 2\pi N_p \tau - 2\arctan\left(-\cos(inc)\tan(\mp \pi N_d \tau)\right).
	\end{eqnarray}
	We know that $\rho_{min} = 0$ if and only if: 
	\begin{equation} \label{eq:condition1}
	1+\cos^2(inc)+\sin^2(inc)\cos(\mp 2\pi N_d \tau) = 0.
	\end{equation} 
	or, alternatively, 
	\begin{equation} \label{eq:condition2}
	\sin\left(\displaystyle\frac{\Delta F}{2}\right) = 0.
	\end{equation}
	
	From the condition in Eq.~\eqref{eq:condition1} a self-intersection happens if:
	\begin{equation}
	\cos(\mp 2\pi N_d \tau) = \displaystyle\frac{-1-\cos^2(inc)}{\sin^2(inc)} = -\frac{2}{\sin^2(inc)} + 1 \leq -1, \quad \forall inc,
	\end{equation}
	which means that polar orbits are the only ones that can fulfill the condition at some point over their relative trajectories. Note that this is a result that we were expecting since polar orbits always have self intersections unless $N_d = 0$, that is, when the rotating frame of reference is coincident with the inertial frame of reference ($\alpha = 0$). 
	
	On the other hand, from the condition in Eq.~\eqref{eq:condition2}:
	\begin{equation}
	2\pi N_p \tau - 2\arctan\left(-\cos(inc)\tan(\mp \pi N_d \tau)\right) = 0 \mod(2\pi),
	\end{equation}
	which can be simplified into:
	\begin{equation}
	\tan(\pi N_p \tau + L\pi) = -\cos(inc)\tan(\mp \pi N_d \tau),
	\end{equation}
	where $L$ is an arbitrary integer number introduced to take into account the modular arithmetic of the equation. However:
	\begin{equation}
	\tan(\pi N_p \tau + L\pi) = \tan(\pi N_p \tau), \quad \forall L\in\mathbb{Z},
	\end{equation}
	and thus, the condition to have a self-intersection in the relative trajectory is:
	\begin{equation} \label{eq:collision_general}
	\cos(inc) = \pm\displaystyle\frac{\tan(\pi N_p \tau)}{\tan(\pi N_d \tau)},
	\end{equation}
	or written as a function of sines and cosines:
	\begin{equation} \label{eq:collision_general2}
	\sin(\pi N_d \tau)\cos(\pi N_p \tau)\cos(inc) = \pm \cos(\pi N_d \tau)\sin(\pi N_p \tau).
	\end{equation}
	
	Let $N_p$ be defined as $N_p = N_d + K$ where $K$ is a given integer which can be either positive or negative. Therefore, the trigonometric functions in $N_p$ from Eq.~\eqref{eq:collision_general2} can be decomposed to obtain:	
	\begin{eqnarray}
	& & \sin(\pi N_d \tau)\cos(\pi N_d \tau)\cos(\pi K \tau)\left(\cos(inc) \mp 1\right) \nonumber \\
	& = & \pm \sin^2(\pi N_d \tau)\sin(\pi K \tau)\left(\cos(inc) \mp 1\right) \pm \sin(\pi K \tau)
	\end{eqnarray}
	which after some algebraic manipulation leads to: 
	\begin{equation}
	\cos(inc) \mp 1 = \pm \displaystyle\frac{\sin(\pi K \tau)}{\sin(\pi N_d \tau)\cos(\pi N_d \tau)\cos(\pi K \tau)-\sin^2(\pi N_d \tau)\sin(\pi K \tau)}.
	\end{equation}
	Additionally, the denominator of the previous expression can be simplified to:
	\begin{eqnarray}
	& & \sin(\pi N_d \tau)\cos(\pi N_d \tau)\cos(\pi K \tau)-\sin^2(\pi N_d \tau)\sin(\pi K \tau) \nonumber \\
	& = & \sin(\pi N_d \tau)\cos(\pi N_d \tau + \pi K \tau) \nonumber \\
	& = & \displaystyle\frac{1}{2}\left[\sin(\pi(N_d-N_p)\tau) + \sin(\pi(N_d+N_p)\tau)\right] \nonumber \\
	& = & \displaystyle\frac{1}{2}\left[-\sin(\pi K \tau) + \sin(\pi(N_d+N_p)\tau)\right],
	\end{eqnarray}
	which reintroduced in the former expression provides the following relation:
	\begin{equation} \label{eq:sine_relation}
	\sin(\pi K \tau) = \displaystyle\frac{\cos(inc) \mp 1}{\cos(inc) \pm 1}\sin(\pi(N_d+N_p)\tau).
	\end{equation}
	In this expression it is important to note that $\tau \in [0,1)$ and also that both sides of the equation have a zero at $\tau = 0$ and at $\tau = 1$ no matter the values of $N_p$, $N_d$ or $K$. This means that both sides of the equation can be equal if and only if $K \neq \{-1, 1\}$, that is, there are additional zeros of $\sin(\pi K \tau)$ in the range $\tau \in [0,1)$. Therefore, $|K| = 1$ is a necessary condition (but not sufficient) to have a relative non-self-intersecting trajectory. In other words, a non-self-intersecting trajectory has either $N_p = N_d - 1$ or $N_p = N_d + 1$. In the following lines we study each case individually.
	
	If $N_p = N_d - 1$, then $K = - 1$, and Eq.~\eqref{eq:sine_relation} becomes:
	\begin{equation}
	\sin(\pi \tau) = -\displaystyle\frac{\cos(inc) \mp 1}{\cos(inc) \pm 1}\sin(\pi(N_d+N_p)\tau),
	\end{equation}
	where it is important to note that the expression:
	\begin{equation}
	-\displaystyle\frac{\cos(inc) \mp 1}{\cos(inc) \pm 1},
	\end{equation}
	is always positive no matter the value of the inclination or direction of the spin rate of the rotating frame of reference. This means that the functions on the left and the right side of the equation are both two sine functions that start with positive values. Additionally, $N_p + N_d = 2N_p + 1 > 1$ ($N_p$ has to be larger than 0 by definition), which means that the frequency on the right side function is always larger than in the left side. Therefore, the only situation where the equation has no solution is when both the maximum value of the left side equation and its slope at $\tau = 0$ are larger than the ones on the right side. That is:
	\begin{eqnarray}
	-\displaystyle\frac{\cos(inc) \mp 1}{\cos(inc) \pm 1} < 1; \nonumber \\
	-\pi(N_d+N_p)\displaystyle\frac{\cos(inc) \mp 1}{\cos(inc) \pm 1} < \pi.
	\end{eqnarray}
	From the first inequality we obtain that $\cos(inc) > 0$ (the orbit must be prograde) for prograde rotating frames ($\alpha > 0$). Conversely, $\cos(inc) < 0$ (the orbit must be retrograde) for retrograde rotating frames ($\alpha < 0$). From the second inequality we can derive that:
	\begin{equation}
	\cos(inc) > \displaystyle\frac{N_d + N_p - 1}{N_d + N_p + 1} = \frac{2N_d - 2}{2N_d} = \frac{N_p}{N_d}.
	\end{equation}
	for prograde rotating reference frames, and:
	\begin{equation}
	\cos(inc) < -\displaystyle\frac{N_d + N_p - 1}{N_d + N_p + 1} = -\frac{2N_d - 2}{2N_d} = -\frac{N_p}{N_d},
	\end{equation}
	for the case of retrograde reference frames.

	If instead $N_p = N_d + 1$, then $K = + 1$, and Eq.~\eqref{eq:sine_relation} becomes:
	\begin{equation} \label{eq:kpositive}
	\sin(\pi \tau) = \displaystyle\frac{\cos(inc) \mp 1}{\cos(inc) \pm 1}\sin(\pi(N_d+N_p)\tau),
	\end{equation}
	where the coefficient:
	\begin{equation}
	\displaystyle\frac{\cos(inc) \mp 1}{\cos(inc) \pm 1},
	\end{equation}
	is now always negative no matter the values of the inclination and $\alpha$ considered. As in the previous case, since the left side of the equation does not have zeros inside the domain $\tau\in(0,1)$, a necessary condition to prevent the equation to have real solution is to impose that the maximum value of the left side of the equation is larger than the right side. In other words:
	\begin{equation}
	1 > \displaystyle\frac{-\cos(inc) \pm 1}{\cos(inc) \pm 1}
	\end{equation} 
	from where it can be derived that $\cos(inc) > 0$ for prograde rotating frames, and $\cos(inc) < 0$ for retrograde frames. Additionally, we know from Eq.~\eqref{eq:collision_general} that the condition for collision is given by:
	\begin{equation}
	\cos(inc) = \pm\displaystyle\frac{\tan(\pi N_p \tau)}{\tan(\pi N_d \tau)},
	\end{equation}
	which means that a self-intersection will occur if and only if:
	\begin{equation}
	\left|\cos(inc)\right| < \left|\max\left(\displaystyle\frac{\tan(\pi N_p \tau)}{\tan(\pi N_d \tau)}\right)\right|.
	\end{equation}
	Therefore, if we combine this condition with the one regarding the sign of $\cos(inc)$ obtained before, we have the following conditions that prevent a relative trajectory to have any self-intersection:
	\begin{eqnarray} \label{eq:open_cos}
	\cos(inc) > \max\left(\displaystyle\frac{\tan(\pi N_p \tau)}{\tan(\pi N_d \tau)}\right), \quad \text{if } \alpha > 0; \nonumber \\
	\cos(inc) < -\max\left(\displaystyle\frac{\tan(\pi N_p \tau)}{\tan(\pi N_d \tau)}\right), \quad \text{if } \alpha < 0;
	\end{eqnarray}
	where $\tau\in[0,1)$. Note that the equation included has no known closed form expression for their maximum value. However, it is possible to better delimit where the maximum of the expression appears.
	
	Equation~\eqref{eq:kpositive} provides us with an equivalent relation to Eq.~\eqref{eq:collision_general} to determine where the maximum of $\tan(\pi N_p \tau)/\tan(\pi N_d \tau)$ appears. Particularly, Eq.~\eqref{eq:kpositive} shows a relation between two functions, $f = \sin(\pi\tau)$ which is always positive in the domain, and:
	\begin{equation}
	g = \displaystyle\frac{\cos(inc) \mp 1}{\cos(inc) \pm 1}\sin(\pi(N_d+N_p)\tau),
	\end{equation}
	which starts being negative, and becomes positive for the first time in $\tau\in[\frac{1}{N_d+N_p},\frac{2}{N_d+N_p}]$. Moreover, we know that the integer $N_d+N_p = 2N_d + 1$ is always an odd number. Therefore:
	\begin{eqnarray}
	\sin(\pi\tau) = \sin(\pi - \pi\tau), \nonumber \\
	\sin(\pi(N_d+N_p)\tau) = \sin(\pi(N_d+N_p) - \pi(N_d+N_p)\tau).
	\end{eqnarray}
	and thus, both functions are symmetrical with respect to $\tau = 1/2$, that is, $f(\tau) = f(\pi - \tau)$, and $g(\tau) = g(\pi - \tau)$.	This implies that if intersections exist between $f$ and $g$ there is at least one intersection in the range $\tau\in[\frac{1}{N_d+N_p},\frac{2}{N_d+N_p}]$ since $f$ is monotonically increasing in $\tau\in[0,1/2)$. Particularly, and due to the mentioned geometry of both functions, the first intersection, if exists, must happen before the maximum of $g$, that is, in $\tau\in[\frac{1}{N_d+N_p},\frac{3}{2}\frac{1}{N_d+N_p}]$. Thus, the maximum of $\tan(\pi N_p \tau)/\tan(\pi N_d \tau)$ happens in the range $\tau\in[\frac{1}{N_d+N_p},\frac{3}{2}\frac{1}{N_d+N_p}]$.
\end{proof}

\begin{figure}[!h]
	\centering
	{\includegraphics[width=\textwidth]{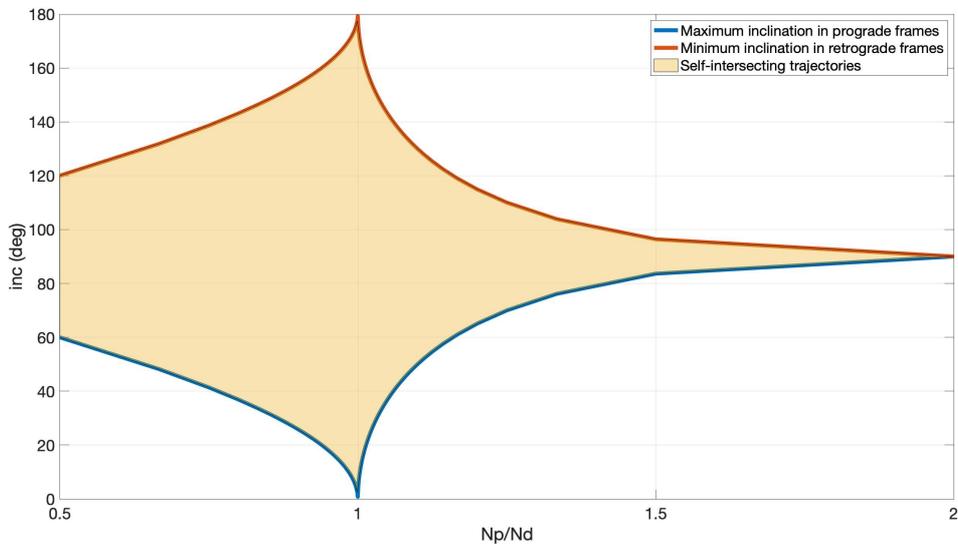}} 
	\caption{Possible inclinations in non-self-intersecting trajectories.}
	\label{fig:inc}
\end{figure}

Figure~\ref{fig:inc} shows the possible relations to obtain non-self-intersecting trajectories when varying the inclination and the rotating frame parameters $N_d$ and $N_p$. Particularly, The upper curve represents the minimum inclination required in retrograde rotating reference frames to obtain a non-self-intersecting trajectory. In a similar manner, the lower curve is the maximum inclination that trajectories defined in a prograde rotating reference frame can present. Finally, the area plotted between these two curves represents the region where it is not possible to define a non-self-intersecting trajectory. 

Moreover, one important thing to note in Theorem~\ref{theorem:nonselfintersectioin} is that when $N_p = N_d + 1$ the condition relating the inclination and the parameters $N_d$ and $N_p$ is not in closed form. To overcome this limitation, an approximated analytical expression is also provided, which presents a maximum absolute error of less than $10^{-3}$ deg in the inclination for any point in the domain.

\begin{lemma} \label{lemma:approx}
	The result from Eq.~\eqref{eq:open_cos} can be approximated by:
	\begin{equation}
	\max\left(\displaystyle\frac{\tan(\pi N_p \tau)}{\tan(\pi N_d \tau)}\right) \approx \frac{\tan\left(N_p\gamma - \displaystyle\frac{N_p\cos(\gamma)}{\left(1+(N_d+N_p)^2\right)\sin(\gamma) -2}\right)}{\tan\left(N_d\gamma - \displaystyle\frac{N_d\cos(\gamma)}{\left(1+(N_d+N_p)^2\right)\sin(\gamma) -2}\right)},
	\end{equation}
	where:
	\begin{equation}
	\gamma = \displaystyle\frac{3}{2}\pi\frac{1}{N_d+N_p}.
	\end{equation}
\end{lemma}

\begin{proof}
	From Theorem~\ref{theorem:nonselfintersectioin} we already know that the maximum of the function $\tan(\pi N_p \tau)/\tan(\pi N_d \tau)$ exists inside the range $\tau\in[\frac{1}{N_d+N_p},\frac{3}{2}\frac{1}{N_d+N_p}]$. In here, it is important to note that the size of the range is reduced as the values of $N_d$ and $N_p$ increase. Therefore it is possible to perform a Taylor expansion in $\tau$ in order to assess where the maximum happens. To that end, and in order to minimize the potential errors of the approximation, it is important to properly select to reference point with respect to which the Taylor expansion is defined. This is done by studying the worst case scenario.
	
	The worst case scenario for the Taylor expansion happens when the range $\tau\in[\frac{1}{N_d+N_p},\frac{3}{2}\frac{1}{N_d+N_p}]$ is maximized, in other words, $N_d$ and $N_p$ are minimized to $N_p = 2$ and $N_d = 1$. Under these conditions, the function to study is:
	\begin{equation}
	\displaystyle\frac{\tan(2\pi \tau)}{\tan(\pi \tau)},
	\end{equation}
	whose first derivative is:
	\begin{equation}
	\displaystyle\frac{\pi(2\sin(\pi\tau)\cos(\pi\tau)-\sin(2\pi\tau)\cos(2\pi\tau))}{\sin^2(\pi\tau)\cos^2(2\pi\tau)} = \frac{\pi\sin(2\pi\tau)(1-\cos(2\pi\tau))}{\sin^2(\pi\tau)\cos^2(2\pi\tau)},
	\end{equation}
	which implies that the maximum of the function is at $\tau = 1/2$, that is, it is on the right boundary of the range $\tau\in[\frac{1}{N_d+N_p},\frac{3}{2}\frac{1}{N_d+N_p}]$. Therefore, we will use the right boundary as the reference point for the expansion.
	
	Let the reference point for the Taylor expansion be:
	\begin{equation}
	\tau_0 = \displaystyle\frac{3}{2}\frac{1}{N_d+N_p}.
	\end{equation}
	Then, the target function can be approximated by:
	\begin{equation}
	f(\tau) = \displaystyle\frac{\tan(\pi N_p \tau)}{\tan(\pi N_d \tau)} \approx f\Big|_{\tau_0} + \frac{df}{d\tau}\Big|_{\tau_0}\left(\tau - \tau_0\right) + \frac{1}{2}\frac{d^2f}{d\tau^2}\Big|_{\tau_0}\left(\tau - \tau_0\right)^2,
	\end{equation}
	where the subindexes represent where the function and its derivatives are evaluated. Since the approximation represents a parabola, we know that the maximum of the previous approximation happens at:
	\begin{equation}
	\tau_m = \tau_0 - \displaystyle\frac{\frac{df}{d\tau}\Big|_{\tau_0}}{\frac{d^2f}{d\tau^2}\Big|_{\tau_0}},
	\end{equation}
	where:
	\begin{eqnarray}
	\displaystyle\frac{df}{d\tau}\Big|_{\tau_0} & = & -\frac{2\pi\cos(\pi\tau_0)}{\left(1+\sin(\pi\tau_0)\right)^2};\nonumber \\
	\displaystyle\frac{df}{d\tau}\Big|_{\tau_0} & = & \frac{2\pi^2\sin(\pi\tau_0)\left(1-(N_d+N_p)^2\right)}{\left(1+\sin(\pi\tau_0)\right)^2} + \frac{4\pi^2\cos^2(\pi\tau)}{\left(1+\sin(\pi\tau_0)\right)^3}.
	\end{eqnarray}
	Note that both expressions are negative, thus the maximum of the function is on the left of $\tau_0$ as derived by Theorem~\ref{theorem:nonselfintersectioin}. Therefore, we can obtain the value of $\tau_m$ just by performing some elemental algebraic manipulations in the previous expressions:
	\begin{equation}
	\tau_m = \tau_0 - \displaystyle\frac{1}{\pi}\frac{\cos(\pi\tau_0)}{\left(1+(N_d+N_p)^2\right)\sin(\pi\tau_0) -2},
	\end{equation}
	This means that the approximated maximum value of function $f$ is:
	\begin{equation}
	\max\left(\displaystyle\frac{\tan(\pi N_p \tau)}{\tan(\pi N_d \tau)}\right) \approx \frac{\tan\left(N_p\gamma - \displaystyle\frac{N_p\cos(\gamma)}{\left(1+(N_d+N_p)^2\right)\sin(\gamma) -2}\right)}{\tan\left(N_d\gamma - \displaystyle\frac{N_d\cos(\gamma)}{\left(1+(N_d+N_p)^2\right)\sin(\gamma) -2}\right)},
	\end{equation}	
	where $\gamma = \pi\tau_0$. 
\end{proof}


\section{Application to satellite constellation design}

The idea behind this satellite constellation design is to generate a closed relative trajectory in a given rotating frame of reference that presents no self-intersections. This means that no matter where satellites are located along this relative trajectory, there is no risk of conjunction. In addition, this methodology allows us to define constellation reconfiguration in a very easy and computationally fast process by re-positioning satellites along the reference relative trajectory. This also allows to assure a minimum safety distance with the rest of the satellites of the constellation even during the reconfiguration process. Finally, this design methodology provides an estimate of the orbital capacity at a given altitude, which is extremely useful when analysing the Space Traffic Management problem.

\subsection{Minimum distance between satellites}

Satellite constellations and slotting architectures distributed in non-self-intersecting trajectories provide a unique opportunity to analytically assess minimum distances between satellites for the overall structure. The reason for this is that if the density of satellites is big enough in this kind of distributions, then it is possible to derive the minimum distance between satellites in the constellation by just computing one of the pairs of the constellation.

\begin{definition}
    A trajectory loop is defined as the section of the relative trajectory comprised in between two consecutive passes over the equator.
\end{definition}

\begin{definition}
    The interloop closest distance is the minimum distance between any two loops in the same hemisphere.
\end{definition}

\begin{figure}[!h]
	\centering
	{\includegraphics[width=0.5\textwidth]{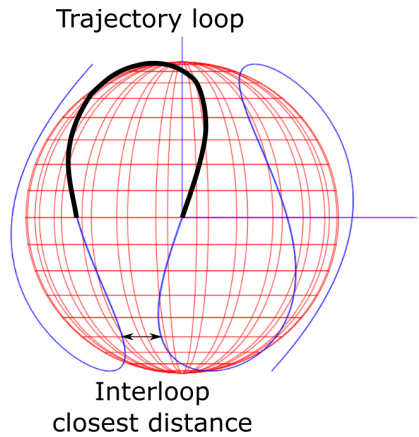}} 
	\caption{Definition of trajectory loop and interloop closest distance.}
	\label{fig:loop}
\end{figure}

It is important to note that for non-self-intersecting relative trajectories, the interloop closest distance happens between trajectory loops that are located in the same hemisphere and that are consecutive to a loop in the opposite hemisphere. As an example of that, Fig.~\ref{fig:loop} shows a trajectory loop (in black) and interloop distance for a non-self-intersecting relative trajectory. Note that this relative trajectory is going to rotate, in general, over time both in the inertial frame of reference and in the Earth-Centered Earth-Fixed frame of reference, however, the trajectory itself is going to be maintained over the dynamics of the system.

\begin{lemma} \label{lemma:mindis}
    If the interloop closest distance is larger than the minimum distance between satellites in a structure, then the overall minimum distance of any pair of satellites in the constellation is:
    \begin{eqnarray}
	    \rho_{min} & = & 2\left|\displaystyle\sqrt{\displaystyle\frac{1+\cos^2(inc)+\sin^2(inc)\cos(\mp 2\pi \frac{N_d}{N_s})}{2}}\sin\left(\frac{\Delta F}{2}\right)\right|; \nonumber \\
	    \Delta F & = & 2\pi \displaystyle\frac{N_p}{N_s} - 2\arctan\left(-\cos(inc)\tan(\mp \pi \frac{N_d}{N_s} )\right),
	\end{eqnarray}
    where $N_s$ is the number of satellite of the constellation. In addition, this minimum distance can be approximated by:
    \begin{equation}
        \rho_{min} \approx 2\pi \displaystyle\frac{N_p}{N_s} + 2\pi\cos(inc)\frac{N_d}{N_s}.
    \end{equation}
\end{lemma}

\begin{proof}
    In a constellation based on a non-self-intersecting trajectory, satellite relative positions are defined combining Eqs.~\eqref{eq:tc_single} and~\eqref{eq:plusminusdistribution}:
    \begin{eqnarray} \label{eq:nsi_distribution}
        \Delta \Omega_{q} & = & \mp 2\pi N_d\displaystyle\frac{q}{N_s}, \nonumber \\
        \Delta M_{q} & = & 2\pi N_p\displaystyle\frac{q}{N_s}.
    \end{eqnarray}
    where $q\in\{1,\dots,N_s\}$ names each satellite of the constellation, and $N_s$ is the total number of spacecraft in the structure. Since all satellites follow the same relative trajectory and a high density of satellites is considered, there are only two possible situations in which the overall minimum distance between two satellites happens. The first situation is generated when two satellites are consecutive in the formation. Conversely, the second situation is generated when satellites are located in different trajectory loops, and thus, the minimum possible distance between satellites that can appear in this case is the interloop closest distance. This means that if the interloop closest distance is larger than the minimum distance between satellites, then only a pair of satellites has to be checked under minimum distance:
    \begin{eqnarray}
        \Delta \Omega & = & \mp 2\pi \displaystyle\frac{N_d}{N_s}, \nonumber \\
        \Delta M & = & 2\pi \displaystyle\frac{N_p}{N_s},
    \end{eqnarray}
    which introduced in Eq.~\eqref{eq:midis} leads to:
    \begin{eqnarray} \label{eq:nsi_mindis}
	    \rho_{min} & = & 2\left|\displaystyle\sqrt{\displaystyle\frac{1+\cos^2(inc)+\sin^2(inc)\cos(\mp 2\pi \frac{N_d}{N_s})}{2}}\sin\left(\frac{\Delta F}{2}\right)\right|; \nonumber \\
	    \Delta F & = & 2\pi \displaystyle\frac{N_p}{N_s} - 2\arctan\left(-\cos(inc)\tan(\mp \pi \frac{N_d}{N_s} )\right).
	\end{eqnarray}
    This expression can be approximated if the number of satellites is large enough, that is, $N_p/N_s \ll 1$ (note that $N_p = N_d \pm 1$ for non-self-intersecting trajectories). This situation is characteristic of megaconstellations and slotting architectures, but there are other space architectures that also meet this condition if $N_p$ is small compared to the number of satellites of the constellation. In this regard, it is important to note that non-self-intersecting trajectories require constrained $N_p$ and $N_d$ values based on the inclination of the orbits. If this condition is met, then a Taylor expansion in $N_p/N_s$ can be performed over Eq.~\eqref{eq:nsi_mindis} to obtain a first order approximation of the expression:
    \begin{equation} \label{eq:mindis_approx}
        \rho_{min} \approx 2\pi \displaystyle\frac{N_p}{N_s} \mp 2\pi\cos(inc)\frac{N_d}{N_s}.
    \end{equation}
\end{proof}

\begin{corollary} \label{cor:ns}
    Under the conditions presented in Lemma~\ref{lemma:mindis}, the approximate maximum number of satellites that can be located in a non-self-intersecting trajectory while maintaining a minimum distance of $\rho_{min}$ is:
    \begin{equation}
        N_s \approx \bigg\lfloor \displaystyle\frac{2\pi}{\rho_{\min}}\left(N_p \mp N_d\cos(inc)\right)\bigg\rfloor.
    \end{equation}
\end{corollary}

\subsection{Example of computation of minimum distance}

In order to provide a deeper insight on Lemma~\ref{lemma:mindis} and to show the performance of the approximation proposed, we include an applied example to a constellation defined in a non-self-intersecting relative trajectory at $60^{\circ}$ in inclination, prograde reference frame and $N_p = 7$. Figure~\ref{fig:min_dis_60} shows the evolution of the minimum distance between satellites as a function of the increase in the number of spacecraft in this specific configuration. This figure clearly shows two distinct regions. The first region is defined in $N_s\in[1,1247]$ and corresponds to the situation where the minimum distance between satellites happens between two spacecraft in two consecutive loops at the same hemisphere. As can be seen, the minimum distance has a large variation when the number of satellites of the constellation is relatively small, and gets more stable as the number of satellites increases. In fact, there is a boundary in the lower value of this variation, which is defined by the interloop closest distance. On the other hand, the second region is defined in $N_s\in[1248,\infty)$ and corresponds to the case in which the minimum distance between satellites happens between two consecutive spacecraft in the relative trajectory. As can be seen, the evolution of the minimum distance in this region is linear in logarithmic scale and maintains that behaviour no matter the number of satellites of the constellation. This is also the region where Lemma~\ref{lemma:mindis} is applied.

\begin{figure}[!h]
	\centering
	{\includegraphics[width=\textwidth]{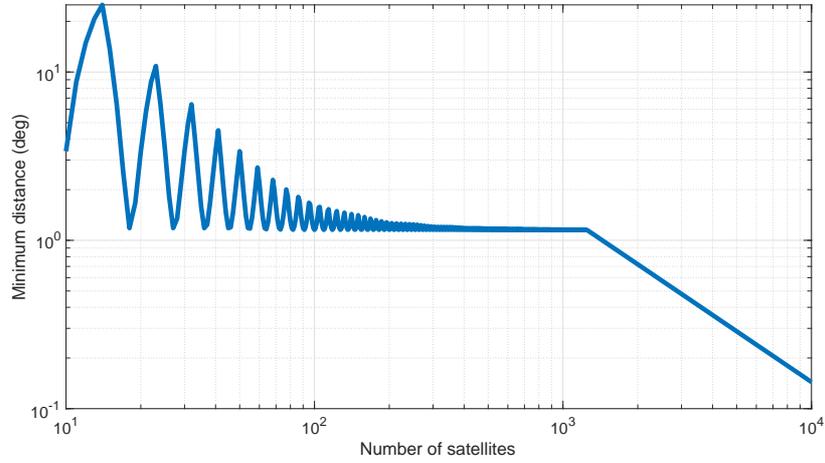}} 
	\caption{Minimum distance between satellites for a non-self-intersecting constellation at $60^{\circ}$ in inclination.}
	\label{fig:min_dis_60}
\end{figure}

\begin{figure}[!h]
	\centering
	{\includegraphics[width=\textwidth]{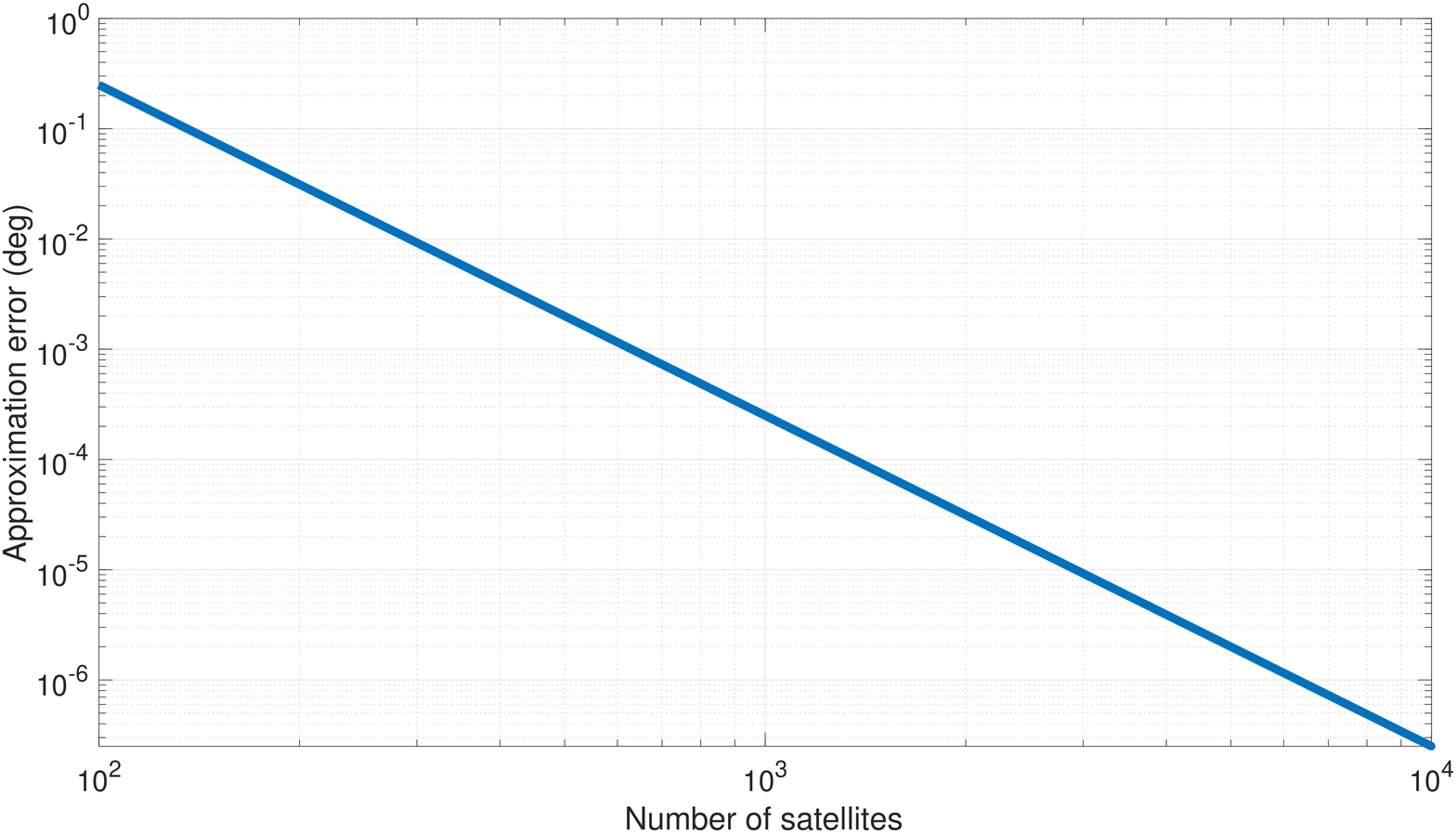}} 
	\caption{Absolute error in minimum distance between satellites by applying Lemma~\ref{lemma:mindis} to a non-self-intersecting constellation at $60^{\circ}$ in inclination.}
	\label{fig:min_dis_error}
\end{figure}

\begin{figure}[!h]
	\centering
	{\includegraphics[width=\textwidth]{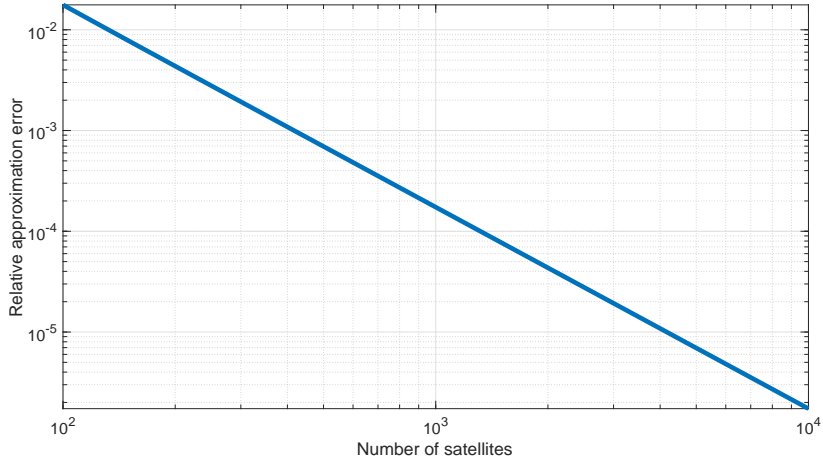}} 
	\caption{Relative error in minimum distance between satellites by applying Lemma~\ref{lemma:mindis} to a non-self-intersecting constellation at $60^{\circ}$ in inclination.}
	\label{fig:mindis_relative_error}
\end{figure}

On the other hand, Fig.~\ref{fig:min_dis_error} shows the error of using the approximation provided by Lemma~\ref{lemma:mindis} when compared with the analytical solution. Note that the figure is defined in all the range of number of satellites, which includes both of the regions described before. However, the approximation and errors presented in the figure only relate to the minimum distance between two consecutive satellites of the constellation, therefore, this approximation is only valid to compute the overall minimum distance of the structure for $N_s\in[1248,\infty)$. The figure shows that the associated error is linear, in logarithmic scale, with the number of satellites of the constellation. Additionally, even for a relatively small number of satellites the error is still smaller than 1 deg. In that regard, it is also interesting to study the relative error of this approximation. This is shown in Fig.~\ref{fig:mindis_relative_error}, where it can be observed that the relative error is two orders of magnitude smaller than the minimum distance when the number of satellites is small, and gets reduced to four and even five orders of magnitude when the constellation is comprised by thousands of satellites. This shows that the approximation can be used effectively to perform the design of slotting architectures and megaconstellations with a good accuracy.

\subsection{Analysis of design possibilities and their capacity}

The previous Lemma~\ref{lemma:mindis} and Corollary~\ref{cor:ns} can be effectively used to analytically define and estimate the capacity of constellations describing non-self-intersecting relative trajectories and use this information to perform effective satellite constellation design. For instance, let a constellation be defined at $60^{\circ}$ in inclination using a non-self-intersecting trajectory to generate the relative positions of the satellites. Since the inclination is $60^{\circ}$, the rotating frame of reference has to be prograde, and thus, $\alpha > 0$. Then, since $N_p/N_d > 0.5$ ($N_p$ and $N_d$ are both integer numbers and $N_p \neq 0$), we can use Theorem~\ref{theorem:nonselfintersectioin} to obtain this relation: 
\begin{equation}
    cos(60^{\circ}) > \max\left(\displaystyle\frac{\tan(\pi N_p \tau)}{\tan(\pi N_d \tau)}\right),
\end{equation}
which is a condition that only happens if and only if $N_p = N_d + 1 \leq 7$.

\begin{figure}[!h]
	\centering
	{\includegraphics[width=\textwidth]{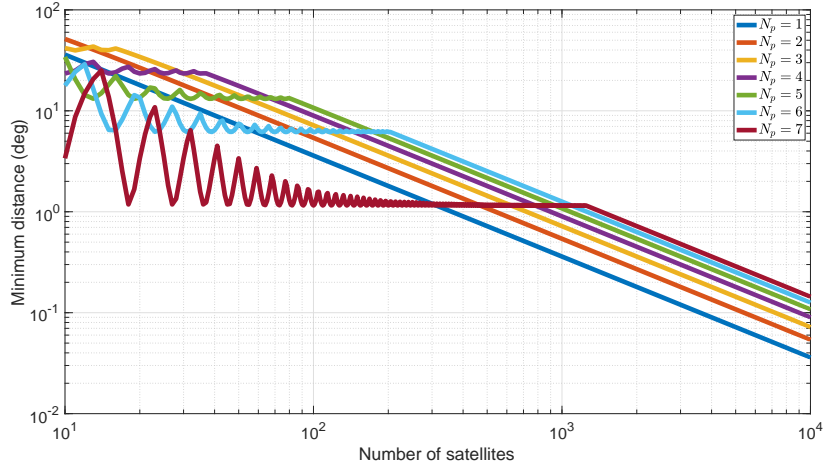}} 
	\caption{Minimum distance between satellites for non-self-intersecting constellations at $60^{\circ}$ in inclination and different rotating frames of reference.}
	\label{fig:NSI_comparison_60}
\end{figure}

Figure~\ref{fig:NSI_comparison_60} shows the evolution of the minimum distance between satellites as a function of the number of spacecraft for all seven possible configurations $N_p=\{1,\dots,7\}$. Note that there are some important properties that can be observed in this figure. First, the lower boundary of these curves in the interloop close approach has a different value, smaller the larger the value $N_p$ is. This happens due to the fact that as $N_p$ increases, so does the number of trajectory loops, and thus, the interloop closest distance has to decrease in consequence. Second, the number of satellites of the constellation in which the function changes its behaviour (when passing from interloop closest distance to distance between consecutive satellites) increases as $N_p$ increases. This situation is produced by the fact that as the interloop closest distance gets smaller, the constellation requires more satellites to reduce the overall minimum distance between spacecraft to match this value. Third, the slope of the curves in the region where the minimum distance is defined by consecutive satellites is approximately the same no matter the particular non-self-intersecting trajectory. Particularly, from Eq.~\eqref{eq:mindis_approx}:
\begin{eqnarray}
    \ln\left( \rho_{min} \right) & \approx & \ln\left(2\pi \displaystyle\frac{N_p}{N_s} \mp 2\pi\cos(inc)\frac{N_d}{N_s}\right) \nonumber \\
    & \approx & \ln\left(2\pi \left(N_p \mp N_d\cos(inc)\right)\right) - \ln\left(N_s\right),
\end{eqnarray}
which represents a straight line with negative slope equal to $-1$, and whose intercept depends on just the inclination and $N_p$ parameter from the relative trajectory. Finally, the case with $N_p = 1$ represents the case in which the relative trajectory is in fact a single inertial orbit that all the satellites of the constellation share. This means that this is the only situation in which only one region is presented in the minimum distance between satellites, and thus, the evolution is linear no matter the number of satellites considered in the structure. 

These properties can be used to select the best non-self-intersecting trajectory in which to base the architecture distribution depending on the expected number of satellites or slots of the space architecture and its orbit inclination. For instance, if we plan to expand a constellation with a large number of satellites, $N_p = 7$ should be the design parameter to choose since it allows to locate more satellites due to the longer trajectory (note also that from Eq.~\eqref{eq:cycle}, $N_p$ is related with the time that a satellite requires to complete this relative trajectory).

\subsection{Estimating the orbital capacity at a given altitude}

The result from Lemma~\ref{lemma:mindis} and Corollary~\ref{cor:ns} can be extended to obtain the estimate of the overall orbital capacity at a given altitude. The main idea is the following. Uniform constellations based on non-self-intersecting trajectories represent a subset of all the possible 2D Lattice Flower Constellations~\cite{nominal,maintenance} for a given number of satellites. Additionally, small changes in inclination and number of satellites have a low impact in the minimum distances between satellites due to their distribution in non-self-intersecting trajectories. This makes these constellations much more resilient to changes in their configuration compared with other 2D Lattice Flower Constellations. For instance, modifying, even if it is by a small quantity, either the number of satellites or the inclination of the orbits in a general 2D Lattice Flower Constellation can produce a large variation in the minimum distance between satellites making an optimal solution to be no longer optimal under this metric. However, constellations based on non-self-intersecting trajectories do not have this variability and thus, they describe the trend evolution of all the 2D Lattice Flower Constellations. In fact, they represent the lower bound in number of satellites for the best possible 2D Lattice Flower Constellations that can be defined for a given value of inclination and number of satellites. Therefore, we can use this property to describe the relation between number of satellites, inclination and minimum distance between satellites. 

\begin{figure}[!h]
	\centering
	{\includegraphics[width=0.8\textwidth]{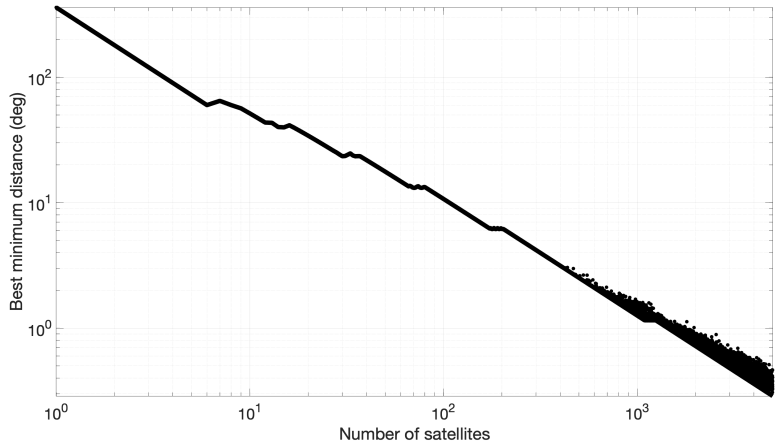}} 
	\caption{Minimum distance between satellites at $60^{\circ}$ in inclination for non-self-intersecting constellations (solid line) and 2D Lattice Flower Constellations (dots).}
	\label{fig:boundary_mindis_5000}
\end{figure}

To that end, we propose a methodology based on the idea of creating all possible combinations of non-self-intersecting trajectories that exist for a given inclination, and selecting the ones with best performance to estimate the orbital capacity of a given shell. This means that for each inclination, we have to obtain all the values of $N_p$ that generate a non-self-intersecting trajectory as seen in Fig.~\ref{fig:NSI_comparison_60} and keep the section of the curves with best performance, that is, largest minimum distance for each value of number of satellites. By doing that, we can obtain a function like the one represented by the solid line in Fig.~\ref{fig:boundary_mindis_5000}. What it is interesting is to compare this result with the best 2D Lattice Flower Constellations that maximize capacity. These optimal constellations can be obtained using the methodology proposed in Ref.~\cite{stm} and are represented in Fig.~\ref{fig:boundary_mindis_5000} as black dots. As can be seen, the solid line matches the results from the best 2D Lattice Flower Constellations up to 400 satellites. After that point, some difference is starting to appear between both solutions, being the difference larger as the number of satellites increases. This difference is generated due to optimal phasing in the orbital intersections that some 2D Lattice Flower Constellations are able to achieve for particular values of inclination, number of satellites and configuration numbers. However, even with this difference, it is possible to observe that constellations based on non-self-intersecting trajectories are still defining the lower bound of all the best 2D Lattice Flower Constellations, and thus defining a very good estimate of the orbital capacity of the system.

It is important to note that, since non-self-intersecting trajectories always exist no matter the inclination of the orbits, constellations defined in these trajectories represent the lower bound in number of satellites for the best possible 2D Lattice Flower Constellations that can be defined for a given value of inclination and number of satellites. Additionally, these non-self-intersecting relative trajectories can be effectively used to reconfigure a given constellation or slotting architecture in an easy a fast process. As an example of that, imagine that we want to distribute $10^5$ satellites at the same altitude and at an inclination of $60^{\circ}$. Using 2D Lattice Flower Constellations this implies checking $246078$ possible combinations. However, using non-self-intersecting trajectories only requires to define the proper $N_p$, and compute the distance between two consecutive satellites using Eq.~\eqref{eq:mindis_approx}, that is, a minimum distance of $0.0144^{\circ}$. This reduces significantly the number of operations required to study these large systems. 

On the other hand, it is also important to note that, as opposed to a general 2D Lattice Flower Constellation, non-self-intersecting distributions are not as sensitive to changes in either the inclination or the orbits, or the relative distribution in mean anomaly when studying the minimum distance between satellites. This means that the structure resultant from this distribution is more resilient to orbital perturbations, which can be useful for many applications.

\subsection{Example of retrograde constellation}

As an example of application of the proposed methodologies, we present in this section an example of retrograde constellation. Particularly, we focus on the optimal 4D Lattice Flower Constellation proposed in Ref.~\cite{4dlfc} that maximizes minimum distance between satellites. This constellation is based on 1000 satellites on sun-synchronous orbits defined at altitudes between $650$ km to $850$ km. This set of constellations have the following distribution parameters: $N_{o} = 500$, $N_{so} = 2$, $N_{c} = 497$, and it was already pointed out in Ref.~\cite{4dlfc} that this set was comprised by constellations defined in non-self-intersecting trajectories. In fact, this is one of the non-self-intersecting trajectories that is not covered by Ref.~\cite{largeconstellations}, but that can be defined with Theorems~\ref{theorem:compatibility} and~\ref{theorem:nonselfintersectioin}. 

From the distribution parameters of this constellation it is possible to obtain the equivalent $N_p$ and $N_d$ values. Particularly, Ref.~\cite{nominal} provided a set of relations to perform this transformation:
\begin{equation}
    N_{so}|N_d, \quad \gcd\left(\displaystyle\frac{N_d}{N_{so}},N_{o}\right) = 1, \quad \text{and} \quad \left(N_dN_{c} -  N_pN_{so}\right)|N_{o}N_{so}.
\end{equation}
which particularized for our specific problem:
\begin{equation}
    2|N_d, \quad \gcd\left(\displaystyle\frac{N_d}{2},500\right) = 1, \quad \text{and} \quad \left(497N_d -  2N_p\right)|1000.
\end{equation}
whose solution is $N_p = -3$, $N_d = 2$. Note that Ref.~\cite{nominal} only studied prograde rotating reference frames, and thus, the reason why $N_p$ is negative for this example. Also, from Eq.~\eqref{eq:ndnp} it is possible to derive that $\alpha = - 3/2 n$, where the mean motion $n$ depends on the altitude of each satellite. This means that the equivalent distribution parameters for the formulation proposed in this work (meaning that both $N_p$ and $N_d$ are positive) are $N_p = 3$, $N_d = 2$. In fact, $N_p = 3$ is the largest value of $N_p$ that can be used to obtain a non-self-intersecting trajectory at these range of sun-synchronous inclinations, which, as said before, maximizes the number of potential satellites that can be located in this kind of constellations. Figure~\ref{fig:Sun_dist} shows the evolution of the minimum distance between satellites for this set of constellations in the range of altitudes between $650$ km and $850$ km. As can be seen, the minimum distance between satellites has a small variation with the inclination, which allows not only to use the same design at multiple altitudes, but also to improve the safety of the system under orbital perturbations. This property was first used in Ref.~\cite{4dlfc} to define a sun-synchronous constellation at five different altitudes using the 4D Lattice Flower Constellation formulation. This allowed to maintain the constellation configuration under the J2 perturbation by imposing that all the relative trajectories were shifting at the same rate. In that regard, it is important to note that non-self-intersecting trajectories do not cover, in general, all the Earth surface at a given instant due to the loop structure of the relative trajectory (see for instance Fig.~\ref{fig:loop}). Nevertheless, it is possible to combine several of these trajectories at slightly different altitudes to prevent this effect as it was shown in Ref.~\cite{4dlfc}.

\begin{figure}[!h]
	\centering
	{\includegraphics[width=0.85\textwidth]{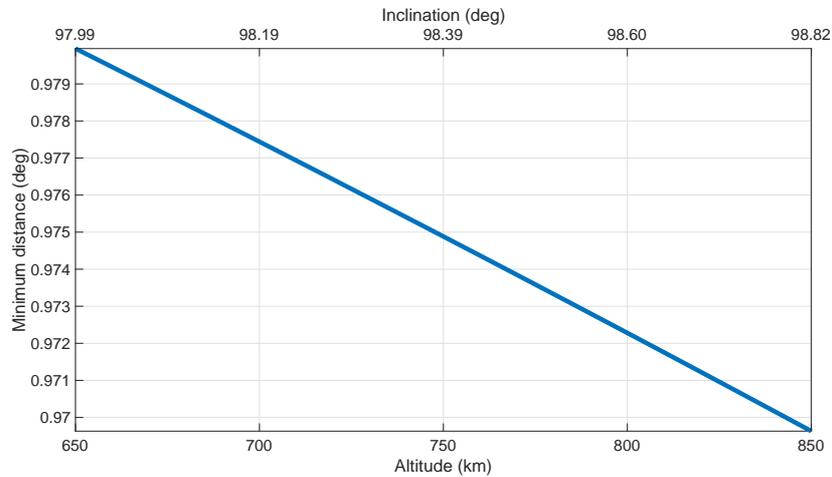}} 
	\caption{Minimum distance between satellites as a function of their altitude for sun-synchronous non-self-intersecting constellations comprised by 1000 spacecrafts.}
	\label{fig:Sun_dist}
\end{figure}


\section{Conclusions}

This work introduces a framework to study orbital capacity in Space Traffic Management problems by the use of non-self-intersecting relative trajectories. This framework is based on the idea of performing the distribution of satellites and slots over relative trajectories that do not present self-intersections in a given rotating frame of reference. This allows to design constellations whose satellites maintain at every moment a minimum distance between other compliant spacecrafts from the structure. Additionally, this work shows that performing the constellation design in these non-self-intersecting relative trajectories allows to have constellations that are more resilient under small variations of the orbital elements and that can be generated and analyzed using closed form analytical expressions. 

To that end, a set of two theorems, two lemmas and a corollary is provided. Particularly, Theorem 1 covers the conditions that two satellites require in order to be placed on the same relative trajectory. Theorem 2 focuses on the definition of the necessary conditions for a relative trajectory not to have self-intersections. On the other hand, Lemma 1 provides a closed expression to estimate the relations provided by Theorem 2. Lemma 2 proposes a closed analytical expression to determine the minimum distance between satellites in a non-self-intersecting relative trajectory. Finally, Corollary 1 provides an approximated relation between the number of satellites of a space structure and the minimum distance between the spacecraft that should be expected.

This theoretical framework is then used to perform satellite constellation design and analyze the orbital capacity of these systems from the point of view of Space Traffic Management. In that regard, we show that non-self-intersecting relative trajectories can be successfully used to design and estimate the capacity of these systems. This reduces drastically the computational time and, at the same time, this framework provides a deeper insight into the dependency of orbital parameters and satellite distributions with the minimum distance between satellites and the capacity of the system.

In addition, satellite constellations based on non-self-intersecting relative trajectories provide several advantages with respect to other satellite distributions. First, they are more resilient to orbital perturbations and other variations in the spacecraft´s orbital parameters. This means that they are safer from a perturbation and uncertainty perspective. Second, no optimization is required to obtain them as they can be fully defined by analytical methods. This allows, for instance, to define the reconfiguration of the whole structure in a very fast process should a break-up event or a collision happens in the vicinity of the architecture. Third, and as a result of the previous two points, these constellations can be used to easily define the expansion of slotting architectures when the technology, the surveillance capabilities, or the control systems allows to reduce either the size of slots or the admissible minimum distance between satellites. Finally, and since they are a subset of 2D Lattice Flower constellations, they define a lower boundary in performance of the best 2D Lattice Constellations that can be found. This allows to use constellations based on non-self-intersecting relative trajectories as a simple reference for other designs. 

\section{Acknowledgements}
This work is sponsored by the Defense Advanced Research Projects Agency (Grant N66001-20-1-4028), and the content of the information does not necessarily reflect the position or the policy of the Government.  No official endorsement should be inferred. 
Distribution statement A: Approved for public release; distribution is unlimited.


\end{document}